\documentclass[orivec]{llncs}

\usepackage{amssymb}
\usepackage{graphicx}
\usepackage{algorithm}
\usepackage{algorithmic}
\usepackage{MnSymbol}
\usepackage{subfig}

\newcommand{\SWITCH}[1]{\STATE \textbf{switch} (#1)}
\newcommand{\ENDSWITCH}{\STATE \textbf{end switch}}
\newcommand{\CASE}[1]{\STATE \textbf{case} (#1)\textbf{:} \begin{ALC@g}}
\newcommand{\ENDCASE}{\end{ALC@g}}

\newcommand{\CASELINE}[2]{\STATE \textbf{case} (#1)\textbf{:} #2}
\newcommand{\IFLINE}[2]{\STATE \textbf{if} #1 \textbf{then} #2}
\newcommand{\ELSIFLINE}[2]{\STATE \textbf{else if} #1 \textbf{then} #2}
\newcommand{\ELSELINE}[1]{\STATE \textbf{else} #1}
\newcommand{\ENDIFLINE}[0]{\STATE \textbf{end if}}
\newcommand{\FORLINE}[2]{\STATE \textbf{for} #1 \textbf{do} #2 \textbf{end for}}

\def\Dcal{\mathcal{D}}
\def\Vcal{\mathcal{V}}

\def\prev{\bullet}
\def\since{\mathrel{\mathbb{S}}}
\def\past{\blacklozenge}
\def\spast{\diamonddot}
\def\impl{\rightarrow}
\def\mn{\mathfrak{m}}

\title{Efficient Runtime Monitoring with Metric Temporal Logic: A Case Study in
the Android Operating System}

\author{Hendra Gunadi\inst{1} \and Alwen Tiu\inst{2}}
\institute{
Research School of Computer Science, The Australian National University \\
\and
School of Computer Engineering,
Nanyang Technological University
}

\begin{document}
\maketitle

\begin{abstract}
We present a design and an implementation of a security policy
specification language based on metric linear-time temporal logic
(MTL). MTL features temporal operators that are indexed by time
intervals, allowing one to specify timing-dependent security policies.
The design of the language is driven by the problem
of runtime monitoring of applications in mobile devices. A main case 
the study is the privilege escalation attack in the Android operating
system, where an app gains access to certain resource or
functionalities that are not explicitly granted to it by the user,
through indirect control flow. To capture these attacks, we extend
MTL with recursive definitions, that are used to express call chains
betwen apps. We then show how the metric operators of MTL, in
combination with recursive definitions, can be used to specify
policies to detect privilege escalation, under various
fine grained constraints. 
We present a new algorithm, extending that of linear time temporal logic, 
for monitoring safety policies written in our specification language. The
monitor does not need to store the entire history of events generated
by the apps, something that is crucial for practical implementations.
We modified the Android OS kernel to allow us to insert our
generated monitors modularly. We have tested the modified OS on an
actual device, and show that it is effective in detecting policy
violations.
\end{abstract}

\section{Introduction}

Android is a popular mobile operating system (OS)  
that has been used in a range of mobile devices such as smartphones and tablet computers. 
It uses Linux as the kernel, which is extended with an application framework (middleware). 
Most applications of Android are written to run on top of this middleware,
and most of Android-specific security mechanisms are enforced at this level. 

Android treats each application as a distinct user with a unique 
user ID. At the kernel level, access control is enforced via the standard Unix permission mechanism
based on the user id (and group id) of the app. 
At the middleware level, each application is sandboxed, i.e., it is running in its own instance of
Dalvik virtual machine, and communication and sharing between apps are
allowed only through an inter-process communication (IPC)
mechanism. Android middleware provides a list of resources and
services such as sending SMS, access to contacts, 
or internet access.
Android enforces access control to these services via its permission mechanism:
each service/resource is associated with a certain unique
permission tag, and each app must request permissions to the
services it needs at installation time. Everytime an app requests access
to a specific service/resource, Android runtime security monitor
checks whether the app has the required permission tags for that
particular service/resource. 
A more detailed discussion of Android security architecture can be found in
\cite{Enck09}.


One 
problem with Android security mechanism 
is the problem of {\em privilege escalation}, that is, the possibility
of an app to gain access to services or resources that it does not have permissions to
access. Obviously privilege escalation is a common problem of every OS, e.g., when
a kernel bug is exploited to gain root access. However, in Android, privilege
escalation is possible even when apps are running in the confine of Android
sandboxes~\cite{defcon18,Davi10,LastPE}. 
There are two types of attacks that can lead to privilege escalation~\cite{LastPE}:
the {\em confused deputy attack} and the {\em collusion attack}. 
In the confused deputy attack, a legitimate app (the deputy) has permissions to
certain services, e.g., sending SMS, and exposes an interface to this functionality
without any guards. This interface can then be exploited by a malicious 
app to send SMS, even though the malicious app does not have the permission. 
Recent studies \cite{defcon18,Grace12NDSS,DroidChecker} 
show some system and consumer apps expose
critical functionalities that can be exploited to launch confused deputy attacks. 
The collusion attack requires two or more malicious apps
to collaborate. 
We have yet to encounter such a malware, either in the Google Play market
or in the third party markets, 
although a proof-of-concept malware with such properties, 
called SoundComber~\cite{Soundcomber}, has been constructed.

Several security extensions to Android have been proposed to deal with privilege
escalation attacks~\cite{QUIRE,IPCInspection,LastPE}. 
Unlike these works, we aim at designing a high-level policy language 
that is expressive enough
to capture privilege escalation attacks, but is also able to express more refined
policies (see Section~\ref{examples}). 
Moreover, we aim at designing a lightweight monitoring framework, where
policy specifications can be modified easily and enforced efficiently. Thus 
we aim at an automated generation of security monitors that can efficiently
enforce policies written in our specification language.

On the specific problem of detecting privilege escalation, it is essentially
a problem of tracking (runtime) control flow, which is in general a difficult
problem and would require a certain amount static analysis~\cite{Denning,TaintDroid}. 
So we adopt a `lightweight' heuristic to ascertain causal dependency between
IPC calls: we consider two successive calls, say from A to B, followed by a call from B to C,
as causally dependent if they happen within a certain reasonably short time frame. 
This heuristic seems sensible in the presence of confused deputy attacks, but
can be of course circumvented by colluding attacks. For the latter there is probably
not a general solution that can be effective in all cases, e.g., 
when covert channels are involved~\cite{Soundcomber}, so we shall restrict to addressing
the former.

The core of our policy language, called RMTL, is essentially a past-fragment 
of metric linear temporal logic (MTL)~\cite{Alur90lics,Thati05MTL,Basin08FSTTCS}.
We consider only the fragment of MTL with past-time operators, as this is sufficient
for our purpose to enforce history-sensitive access control. This also means that
we can only express safety properties~\cite{Lichtenstein}, but not policies capturing
obligations as in, e.g., \cite{Basin08FSTTCS}. 
Temporal operators are useful in this setting to 
enforce access control on apps based on histories of their executions;
see Section~\ref{examples}.  
Such a history-dependent policy cannot be expressed in
the policy languages used in \cite{QUIRE,IPCInspection,LastPE}.

MTL by itelf is, however, insufficient to express transitive closures 
of relations, which is needed to specify an IPC call chains between apps, among
others. 
To deal with this, we extend MTL with a recursive definitions, e.g., one
would be able to write a definition such as:
\begin{equation}
\label{eq:trans}
trans(x,y) := call(x,y) \lor \exists z. \diamonddot_n trans(x,z) \land call(z,y), 
\end{equation}
where $call$ denotes the IPC event. This equation 
defines $trans$ the reflexive-transitive closure of $call.$ The metric
operator $\diamonddot_n \phi$ means intuitively $\phi$ holds within $n$
time units in the past; we shall see a more precise definition of the 
operators in Section~\ref{logic}. 
Readers familiar with modal $\mu$-calculus~\cite{Bradfield07handbook} 
will note that this is but a syntactic sugar for 
$\mu$-expressions for (least) fixed points.


To be practically enforceable in Android, RMTL monitoring algorithm must satisfy an
important constraint, i.e., the algorithm must be {\em trace-length independent}. 
This is because the number of events generated by Android can range in the thousands
per hour, so if the monitor must keep all the events generated by
Android, its performance will degrade significantly overtime. 
Another practical consideration also motivates a restriction to metric operators
that we adopt in RMTL. More specifically, MTL allows a metric version
of the `since' operator of the form $\phi_1 \since_{[m,n)} \phi_2$, where
$[m, n)$ specifies a half-closed (discrete) time interval from $m$ to $n$. 
The monitoring algorithm for MTL in \cite{Thati05MTL} works by first expanding this
formula into formulas of the form $\phi_1 \since_{[m',n')} \phi_2$ where
$[m',n')$ is any subinterval of $[m,n)$. A similar expansion is also used 
implicitly in monitoring for first-order MTL in \cite{Basin08FSTTCS}, i.e., in their
incremental automatic structure extension in their first-order logic encoding for
the `since' and `until' operators. In general, if we have $k$ 
nested occurrences of metric operators, each with interval $[m,n)$, the number of
formulas produced by this expansion is bounded by $O((\frac{(n-m) \times (n - m + 1)}{2})^{k})$. 
In Android, event timestamps are in milliseconds, so this kind of expansion
is not practically feasible. 
For example, suppose we have a policy that monitors three successive
IPC calls that happen within 10 seconds between successive calls. 
This requires two nested metric operators with intervals $[0,10^4)$
to specify. The above naive expansion would produce around
$25 \times 10^{14}$ formulas, and assuming the truth value of each formula
is represented with 1 bit, this would require more than $30$ GB of 
storage to store all their truth values,
something which is beyond the capacity of most smartphones today. 
An improvement to this exponential expansion is proposed in \cite{Basin11RV,Reinbacher12RV}, where one
keeps a sequence of timestamps for each metric temporal operator occuring
in the policy. 
This solution, although avoids the exponential
expansion, is strictly speaking not trace-length independent. 
This solution seems optimal so it is hard to improve it without further
restriction to the policy language. We show that, if one restricts
the intervals of metric operators to the form $[0,n)$, one only needs to
keep one timestamp for each metric operator in monitoring; see Section~\ref{monitor}.
 


To summarise, our contributions are as follows:
\begin{enumerate}
\item In terms of results in runtime verification, 
our contribution is in the design of a new logic-based policy 
language that extends MTL with recursive definitions, that avoids
exponential expansion of metric operators, and for
which the policy enforcement is trace-length independent.

\item In terms of the application domain, ours is the first 
implementation of a logic-based runtime security monitor for Android 
that can enforce 
history-based access control policies, including those that concern 
privilege escalations. Our monitoring framework can express temporal and metric-based
policies not possible in existing works \cite{QUIRE,IPCInspection,LastPE,bauer2013rv}. 

\end{enumerate}

The rest of the paper is organized as
follows. 
Section~\ref{logic} introduces our policy language RMTL. 
Some example policies are described in \ref{examples}. 
Section~\ref{implementation} discusses our implementation of the monitors
for RMTL, and the required modification of Android OS kernel
to integrate our monitor into the OS. 
In Section~\ref{conclusion} we conclude and discuss related and future works.

Some proofs of the lemmas and theorems can be found in the appendix. 
Details of the implementation of the
monitor generator and the binaries of the modified
Android OS are available online.\footnote{
\url{http://users.cecs.anu.edu.au/\~{}hengunadi}/LogicDroid.html.}

\paragraph{Acknowledgment} This work is partly supported by the Australian Research Council Discovery Grant DP110103173.

\section{The policy specification language RMTL}
\label{logic}

Our policy specification language, which we call RMTL, is based on an extension of 
metric linear temporal logic (MTL)~\cite{MTL}. The semantics of LTL~\cite{Pnueli77} 
is defined in terms of models which are sequences of states (or worlds). 
In our case, we restrict to finite sequences of states. 
MTL extends LTL models by adding timestamps to each state,
and adding temporal operators that incorporate timing constraints, e.g.,
MTL features temporal operators such as $\diamond_{[0,3)} \phi$ which
expresses that $\phi$ holds in some state in the future, and the timestamp
of that world is within 0 to 3 time units from the current timestamp. 
We restrict to a model of MTL that uses discrete time, i.e.,
timestamps in this case are non-negative integers. 
We shall also restrict to the past-time fragment of MTL.

We extend MTL with two additional features: first-order quantifiers and recursive
definitions. Our first-order language is a multi-sorted language. For this paper,
we consider only two sorts, which we call $prop$ (for `properties') and $app$
(for denoting applications). Sorts are ranged over by $\alpha$. 
We first fix a {\em signature} $\Sigma$ for our first-order language, which is used
to express terms and predicates of the language. 
We shall consider only constant symbols and predicate symbols, but no function symbols.
We distinguish two types of predicate symbols: {\em defined} predicates 
and {\em undefined} ones.
The defined predicate symbols are used to write recursive definitions and 
to each of such symbols we associate
a formula as its definition. We shall come back to this shortly. 

Constant symbols are ranged over by $a$, $b$ and $c$,
undefined predicate symbols are ranged over by $p$, $q$ and $r$, and defined predicate
symbols are ranged over by $P$, $Q$ and $R.$ 
We assume an infinite set of sorted variables $\Vcal$,
whose elements are ranged over by $x$, $y$ and $z.$ We sometimes write $x_\alpha$ 
to say that $\alpha$ is the sort of variable $x.$
A {\em $\Sigma$-term} is either a constant symbol 
$c\in\Sigma$ or a variable $x\in \Vcal$.  We use $s$, $t$ and $u$ to range over terms. 
To each symbol in $\Sigma$ we associate a sort information. We shall write $c: \alpha$
when $c$ is a constant symbol of sort $\alpha.$ A predicate symbol of arity $n$ has
sort of the form $\alpha_1 \times \cdots \times \alpha_n$, and such a predicate
can only be applied to terms of sorts $\alpha_1,\dots,\alpha_n.$

Constant symbols are used to express things like permissions in the Android OS, e.g.,
reading contacts, sending SMS, etc., and user ids of apps. 
Predicate symbols are used express
events such as IPC calls between apps, and properties of an app, such as whether
it is a system app, a trusted app (as determined by the user). 
As standard in first-order logic (see e.g. \cite{fitting}), the semantics of terms
and predicates are given in terms of a first-order structure, i.e., a set $\Dcal_\alpha$,
called a {\em domain}, for each sort $\alpha$, and an interpretation
function $I$ assigning each constant symbol $c : \alpha \in \Sigma$ an element 
of $c^I \in \Dcal_\alpha$ and each predicate symbol 
$p : \alpha_1 \times \cdots \times \alpha_n \in \Sigma$ 
an $n$-ary relation $p^I \subseteq  \Dcal_{\alpha_1} \times \cdots \times \Dcal_{\alpha_n}.$
We shall assume constant domains in our model, i.e., every world has the same domain. 

The formulas of RMTL is defined via the following grammar:
$$
\begin{array}{ll}
F := & \bot \mid p(t_1,\dots,t_m) \mid P(t_1,\dots,t_n) \mid F \lor F \mid 
\neg F \mid \prev F \mid F \since F \mid \past F \mid \spast F \mid \\
& \prev_{n} F \mid F \since_{n} F \mid \past_{n} F \mid \spast_{n} F \mid 
\exists_\alpha x.F\\
\end{array}
$$
where $m$ and $n$ are natural numbers. The existential quantifier is annotated
with a sort information $\alpha$. For most of our examples and applications, 
we only quantify only over variables of sort $app$. 
The operators indexed by $n$ are {\em metric temporal operators}.
The $n \geq 1$ here denotes the interval $[0,n)$, so these are special cases
of the more general MTL operators in \cite{MTL}, where intervals can take 
the form $[m,n)$, for $n \geq m \geq 0.$ 
We use $\phi$, $\varphi$ and $\psi$ to range over formulas. We assume that unary 
operators bind stronger than the binary operators, so $\prev \phi \lor \psi$
means $(\prev \phi) \lor \psi$.
We write $\phi(x_1,\dots,x_n)$ to denote a formula whose free variables are
among $x_1,\dots,x_n$. Given such a formula, we write $\phi(t_1,\dots,t_n)$
to denote the formula obtained by replacing $x_i$ with $t_i$ for every
$i \in \{1,\dots,n\}$.

To each defined predicate symbol $P : \alpha_1 \times \cdots \times \alpha_n$, 
we associate a formula $\phi_P$, which we call the {\em definition} of $P$. 
Notationally, we write $P(x_1,\dots,x_n) := \phi_p(x_1,\dots,x_n).$ 
We require that $\phi_P$ is {\em guarded}, i.e., every occurrence of any
recursive predicate $Q$ in $\phi_P$ is prefixed by 
either $\bullet$, $\bullet_m$, $\diamonddot$ or $\diamonddot_n$. 
This guardedness condition is important to guarantee termination of recursion in
model checking. 

Given the above logical operators, we can define additional operators via their negation,
e.g., $\top$ is defined as $\neg \bot$, $\phi \land \psi$ is defined as $\neg (\neg \phi \lor \neg \psi)$, 
$\phi \impl \psi$ is defined as $\neg \phi \lor \psi$, and 
$\forall_\alpha x.\phi$ is defined as $\neg (\exists_\alpha x.\neg \phi)$, etc. 

Before proceeding to the semantics of RMTL, we first define a well-founded ordering on
formulae of RMTL, which will be used later.
\begin{definition}
\label{def:order}
We define a relation $\prec_S$ on the set RMTL formulae as the smallest
relation satisfying the following conditions:
\begin{enumerate}
\item For any formula $\phi$ of the form
$p(\vec t)$, $\bot$, 
$\bullet \psi$, $\bullet_n \psi$, $\diamonddot \psi$ and $\diamonddot_n \psi$,
there is no $\phi'$ such that $\phi' \prec_S \phi$. 
\item For every recursive definition $P(\vec x) := \phi_P(\vec x)$,
we have $\phi_P(\vec t) \prec_S P(\vec t)$ for every terms $\vec t.$
\item $\psi \prec_S \psi \lor \psi'$, $\psi \prec_S \psi' \lor \psi$,
$\psi \prec_S \neg \psi$, and $\psi \prec_S \exists x.\psi$.
\item $\psi_i \prec_S \psi_1 \since \psi_2$, and $\psi_i \prec_S \psi_1 \since_n \psi_2$, 
for $i \in \{1,2\}$
\end{enumerate}
We denote with $\prec$ the reflexive and transitive closure of $\prec_S.$
\end{definition}

\begin{lemma}
\label{lm:order}
The relation $\prec$ on RMTL formulas is a well-founded partial order. 
\end{lemma}

For our application, we shall restrict to finite domains. 
Moreover, we shall restrict to an interpretation $I$
which is injective, i.e., mapping every constant $c$ to a unique element of $\Dcal_{\alpha}.$
In effect we shall be working in the term model, so elements of $\Dcal_\alpha$ are
just constant symbols from $\Sigma.$ So we shall use a constant symbol, say $c : \alpha$,
to mean both $c \in \Sigma$ and $c^I \in \Dcal_\alpha$. 
With this fix interpretation, the definition of
the semantics (i.e., the satisfiability relation) can be much simplified, e.g., we do not
need to consider valuations of variables. 
A {\em state}  is a set of undefined atomic formulas of the form
$p(c_1,\dots,c_n)$. 
Given a sequence $\sigma$, we write $|\sigma|$ to denote its length, and 
we write $\sigma_i$ to denote the $i$-th element of $\sigma$
when it is defined, i.e., when $1 \leq i \leq |\sigma|.$ 
A {\em model} is a pair $(\pi,\tau)$ of a sequence of {\em states} $\pi$ and a sequence 
{\em timestamps}, which are natural numbers, such that $|\pi| = |\tau|$ and
$\tau_i \leq \tau_j$ whenever $i\leq j$. 

Let $<$ denote the total order on natural numbers. Then we can define 
a well-order on pairs $(i,\phi)$ of natural numbers and formulas by taking the lexicographical
ordering $(<, \prec)$. 
The satisfiability relation between a model $\rho = (\pi,\tau)$, a {\em world} $i \geq 1$ (which is a natural number) 
and a {\em closed} formula $\phi$ (i.e., $\phi$ contains no free variables),
written $(\rho,i) \models \phi$, is defined by induction on the pair $(i, \phi)$ as follows,  
where we write $(\rho, i) \not \models \phi$ when $(\rho, i) \models \phi$ is false. 
\begin{itemize}
\item $(\rho, i) \not \models \bot$ 
\item $(\rho, i) \models \neg \phi$ iff $(\rho, i) \not \models \phi$.
\item $(\rho, i) \models p(c_1,\dots,c_n)$ iff $p(c_1,\dots,c_n) \in \pi_i.$
\item $(\rho, i) \models P(c_1,\dots,c_n)$ iff $(\rho,i) \models \phi(c_1,\dots,c_n)$
where $P(\vec x) := \phi(\vec x).$
\item $(\rho, i) \models \phi \lor \psi$ iff 
$(\rho, i) \models \phi$ or $(\rho,i) \models \psi$. 
\item $(\rho, i) \models \prev \phi$ iff $i > 1$ and $(\rho, i-1) \models \phi.$
\item $(\rho, i) \models \past \phi$ iff there exists $j \leq i$ s.t. $(\rho, j) \models \phi$. 
\item $(\rho, i) \models \diamonddot \phi$ iff $i > 1$ and there exists $j < i$ s.t. $(\rho, j) \models \phi.$
\item $(\rho, i) \models \phi_1 \since \phi_2$ 
iff there exists $j \leq i$ such that $(\rho, j) \models \phi_2$
and $(\rho,k)\models \phi_1$ for every $k$ s.t. $j < k \leq i.$ 

\item $(\rho, i) \models \prev_n \phi$ iff $i > 1$, $(\rho, i-1) \models \phi$ and $\tau_i - \tau_{i-1} < n.$
\item $(\rho, i) \models \past_n \phi$ iff there exists $j \leq i$ s.t. $(\rho, j) \models \phi$ and $\tau_i - \tau_j < n$. 
\item $(\rho, i) \models \diamonddot_n \phi$ iff $i > 1$ and there exists $j < i$ s.t. $(\rho, j) \models \phi$ and $\tau_i - \tau_j < n.$
\item $(\rho, i) \models \phi_1 \since_n \phi_2$ 
iff there exists $j \leq i$ such that $(\rho, j) \models \phi_2$, 
$(\rho,k)\models \phi_1$ for every $k$ s.t. $j < k \leq i$, and $\tau_i - \tau_j < n.$ 
\item $(\rho, i) \models \exists_\alpha x.\phi(x)$ iff 
there exists $c \in \Dcal_\alpha$ s.t. $(\rho, i) \models \phi(c).$
\end{itemize}
Note that due to the guardedness condition in recursive definitions, 
our semantics for recursive predicates is much simpler than the usual definition
as in $\mu$-calculus, which typically involves the construction of a (semantic) fixed point operator. 
Note also that some operators are redundant, e.g., $\past \phi$ can be defined as
$\top \since \phi$, and $\diamonddot \phi$ can be defined
as $\prev \past \phi.$ This holds for some metric operators, e.g., 
$\past_n \phi$ and $\diamonddot_n \phi$ can be defined as,
respectively, $\top \since_n \phi$ and 
\begin{equation}
\label{eq:diamonddot}
\diamonddot_n \phi = \bigvee_{i+j = n} \prev_i \past_j \phi
\end{equation}
This operator will be used to specify an active call chain, as we shall see later, so
it is convenient to include it in our policy language.

In the next section, we shall assume that $\past$, $\diamonddot$, 
$\past_n$ as derived connectives. 
Since we consider only finite domains, $\exists_\alpha x.\phi(x)$ can be
reduced to a big disjunction $\bigvee_{c \in \Dcal_\alpha} \phi(c)$, so we shall
not treat $\exists$-quantifier explicitly. 
This can be problematic if the domain of quantification is big, 
as it suffers the same kind
of exponential explosion as with the expansion of metric operators in MTL~\cite{Thati05MTL}. 
We shall defer the explicit treatment of quantifiers to future work.

\section{Trace-length independent monitoring}
\label{monitor}

The problem of monitoring is essentially a problem of model checking, i.e., to decide
whether $(\rho, i) \models \phi$, for any given $\rho = (\pi,\tau)$, $i$ and $\phi.$
In the context of Android runtime monitoring, a state in $\pi$ can be any events
of interest that one would like to capture, e.g., the IPC call events, 
queries related to location information or contacts, etc.
To simplify discussions, and because our main interest is
in privilege escalation through IPC, the only type of event we consider
in $\pi$ is the IPC event, which we model with the 
predicate $call : app \times app$. 

Given a policy specification $\phi$, a naive monitoring algorithm that enforces this
policy would store the entire event history $\pi$ and every time a new event
arrives at time $t$, it would check $(([\pi; e], [\tau; t]), |\rho| + 1) \models \phi.$
This is easily shown decidable, but is of course rather inefficient.  In general, the model
checking problem for RMTL (with finite domains) can be shown PSPACE hard 
following the same argument as in \cite{FOTL}.
A design criteria of RMTL is that enforcement of
policies does not depend on the length of history of events, i.e., at any time 
the monitor only needs to keep track of a fixed number of states. 
Following~\cite{bauer2013rv}, we call a monitoring algorithm that satisfies 
this property {\em trace-length independent.}

For PTLTL, trace-length independent monitoring algorithm exists, e.g., 
the algorithm by Havelund and Rosu~\cite{DYNAMICLTL}, which depends only on
two states in a history. That is, satisfiability of $(\rho, i+1) \models \phi$
is a boolean function of satisfiability of $(\rho, i + 1) \models \psi$, for every
strict subformula $\psi$ of $\phi$, and satisfiability of $(\rho, i) \models \psi'$,
for every subformula $\psi'$ of $\phi.$ 
This works for PTLTL because the semantics of temporal operators in PTLTL 
can be expressed in a recursive form, e.g., the semantics of $\since$ can
be equally expressed as~\cite{DYNAMICLTL}: 
$(\rho, i+1) \models \phi_1 \since \phi_2$ iff
$(\rho, i+1) \models \phi_2$, or
$(\rho, i + 1) \models \phi_1$ and 
$(\rho, i) \models \phi_1 \since \phi_2.$
This is not the case for MTL. 
For example, satisfiability of 
the unrestricted `since' operator $\since_{[m,n)}$ can be equivalently expressed as:
\begin{equation}
\begin{array}{ll}
(\rho, i+1) \models \phi_1 \since_{[m,n)} \phi_2  \mbox{ iff  } & m = 0, n > 1, \mbox{ and } (\rho, i+1) \models \phi_2, \mbox{ or  } \\
 & (\rho, i+1) \models \phi_1 \mbox{ and } (\rho, i) \models \phi_1 \since_{[m', n')} \phi_2
\end{array}
\label{eq:since}
\end{equation}
where $m' = min(0, m - \tau_{i+1} + \tau_i)$ and $n' = min(0, n - \tau_{i+1} + \tau_i).$
Since $\tau_{i+1}$ can vary, the value of $m'$ and $n'$ can vary, depending on
the history $\rho.$ 
We avoid the expansion of metric operators in monitoring by 
restricting the intervals in the metric operators to the form $[0, n).$
We show that clause (\ref{eq:since}) can be brought back to a purely recursive form. 
The key to this is the following lemma:

\begin{lemma}[Minimality]
\label{lm:minimal}
If $(\rho, i) \models \phi_1 \since_n \phi_2$ ($(\rho, i) \models \diamonddot_n \phi$) then there exists an $m \leq n$ 
such that 
$(\rho, i) \models \phi_1 \since_m \phi_2$ (resp. $(\rho, i) \models \diamonddot_m \phi$) and 
such that for every $k$ such that $0 < k < m$, 
we have $(\rho, i) \not \models \phi_1 \since_k \phi_2$ (resp., 
$(\rho, i) \not \models \diamonddot_k \phi$).
\end{lemma}
Given $\rho$, $i$ and $\phi$, we define a function $\mn$ as follows:
$$
\mn(\rho,i,\phi) = 
\left\{
\begin{array}{l}
m, \mbox{ if $\phi$ is either $\phi_1 \since_n \phi_2$
or $\diamonddot_n \phi'$ and $(\rho,i) \models \phi$, } \\
0, \mbox{ otherwise. }
\end{array}
\right. 
$$
where $m$ is as given in Lemma~\ref{lm:minimal}; we shall see how its value
is calculated in Algorithm~\ref{Iter_Algorithm}.  
The following theorem follows from Lemma~\ref{lm:minimal}.
\begin{theorem}[Recursive forms]
\label{thm:recursive-form}
For every model $\rho$, every $n \geq 1$, $\phi$, $\phi_1$ and $\phi_2$, 
and every $1 < i \leq |\rho|$, the following hold:
\begin{enumerate}
\item $(\rho, i) \models \phi_1 \since_n \phi_2$ iff
$(\rho, i) \models \phi_2$, or 
$(\rho, i) \models \phi_1$ and 
$(\rho, i-1) \models \phi_1 \since_n \phi_2$ and 
$n - (\tau_i - \tau_{i-1}) \geq \mn(\rho, i-1, \phi_1 \since_n \phi_2).$

\item $(\rho, i) \models \diamonddot_n \phi$ iff
$(\rho, i - 1) \models \phi$ and $\tau_i - \tau_{i-1} < n$, or 
$(\rho, i-1) \models \diamonddot_n \phi$ and 
$n - (\tau_i - \tau_{i-1}) \geq \mn(\rho, i-1, \diamonddot_n \phi).$

\end{enumerate}
\end{theorem}

Given Theorem~\ref{thm:recursive-form}, the monitoring algorithm for PTLTL in \cite{DYNAMICLTL}
can be adapted, but with an added data structure to keep track of the function $\mn.$ 
In the following, given a formula $\phi$, we assume that $\exists$, $\past$ and $\diamonddot$ 
have been replaced with its equivalent form as mentioned in Section~\ref{logic}. 

Given a formula $\phi$, let $Sub(\phi)$ be the set of subformulas of $\phi$. We define
a closure set $S^*(\phi)$ of $\phi$ as follows:  Let $Sub^0(\phi) = Sub(\phi)$, and let 
$$
Sub^{n+1}(\phi) = Sub_n(\phi) \cup 
\{Sub(\phi_P(\vec c)) \mid P(\vec c) \in Sub_n(\phi), \mbox{ and } P(\vec x) := \phi_P(\vec x)\}
$$
and define $Sub^*(\phi) = \bigcup_{n\geq 0} Sub^n(\phi).$ 
Since $\Dcal_\alpha$ is finite,
$Sub^*(\phi)$ is finite, although its size is exponential in the
arities of recursive predicates. For our specific applications, the predicates used in our
sample policies have at most arity of two 
(for tracking transitive calls), so this is still tractable. In future work, 
we plan to investigate ways of avoiding this explicit expansion
of recursive predicates.

We now describe how monitoring can be done for $\phi$, given
$\rho$ and $1 \leq i \leq |\rho|.$ 
We assume implicitly a preprocessing step where we compute $Sub^*(\phi)$; we do not describe
this step here but it is quite straightforward. 
Let $\phi_1, \phi_2, \ldots, \phi_m$ be an enumeration of $Sub^*(\phi)$
respecting the partial order $\prec$, i.e., if $\phi_i \prec \phi_j$
then $i \leq j.$
Then we can assign to each $\psi \in Sub^*(\phi)$ an index $i$, s.t., $\psi = \phi_i$ in this 
enumeration. We refer to this index as $idx(\psi).$ 
We maintain two boolean arrays $prev[1,\dots,m]$ and $cur[1,\dots,m].$
The intention is that given $\rho$ and $i > 1$, 
the value of $prev[k]$ corresponds to the truth value of the judgment
$(\rho, i-1) \models \phi_k$ and the truth value of 
$cur[k]$ corresponds to the truth value of the judgment 
$(\rho, i) \models \phi_k.$ 
We also maintain two integer arrays $mprev[1,\dots,m]$ and $mcur[1,\dots,m]$
to store the value of the function $\mn.$ The value of $mprev[k]$
corresponds to $\mn(\rho, i-1, \phi_k)$, and $mcur[k]$ corresponds
to $\mn(\rho, i, \phi_k).$
Note that this preprocessing step
only needs to be done once, i.e., when generating the monitor codes for
a particular policy, which is done offline, prior to inserting the monitor into the
operating system kernel.

The main monitoring algorithm 
is divided into two subprocedures:
the initialisation procedure (Algorithm~\ref{Init_Algorithm})
and the iterative procedure (Algorithm~\ref{Iter_Algorithm}).
The monitoring procedure (Algorithm~\ref{Monitor_Algorithm}) is then a simple combination of these two.
We overload some logical symbols to denote operators on boolean values. 
In the actual implementation, we do not actually implement the loop in Algorithm~\ref{Monitor_Algorithm};
rather it is implemented as an event-triggered procedure, to process each event as they arrive using $Iter$.  

{\small
\begin{algorithm}[t]
\caption{$Monitor(\rho, i, \phi)$}
\label{Monitor_Algorithm}
\begin{algorithmic}
\STATE $Init(\rho,\phi,prev,cur,mprev,mcur)$
\FOR {$j = 1$ to $i$}

\STATE $Iter(\rho,j,\phi,prev,cur,mprev,mcur);$
\ENDFOR

\RETURN $cur[idx(\phi)];$
\end{algorithmic}
\end{algorithm}
}

{\small
\begin{algorithm}[t]
\caption{$Init(\rho, \phi, prev, cur, mprev, mcur)$} 
\label{Init_Algorithm}
\begin{algorithmic}
\FOR {$k=1,\dots,m$}
\STATE $prev[k] := false$, $mprev[k] := 0$ and $mcur[k] := 0;$
\ENDFOR
\FOR {$k=1,\dots,m$}

\SWITCH {$\phi_k = \bot$}
\CASELINE{$\bot$}{$cur[k] := false;$}

\CASELINE {$p(\vec c)$}{$cur[k] := p(\vec c) \in \pi_1;$}

\CASELINE {$P(\vec c)$}{$cur[k] := cur[idx(\phi_P(\vec c))];$ \COMMENT{Suppose $P(\vec x) := \phi_P(\vec x).$}}

\CASELINE {$\neg \psi$}{$cur[k] := \neg cur[idx(\psi)];$}

\CASELINE {$\psi_1 \lor \psi_2$}{$cur[k] := cur[idx(\psi_1)] \lor cur[idx(\psi_2)];$}

\CASELINE {$\prev \psi$}{$cur[k] := false;$}

\CASELINE {$\diamonddot \psi$}{$cur[k] := false;$}

\CASELINE {$\psi_1 \since \psi_2$}{$cur[k] := cur[idx(\psi_2)];$}

\CASELINE {$\prev_n \psi$}{$cur[k] := false;$}

\CASELINE {$\diamonddot_n \psi$}{$cur[k] := false; mcur[k] := 0;$}

\CASE {$\phi_k = \psi_1 \since_n \psi_2$}
\STATE $cur[k] := cur[idx(\psi_2)];$
\IFLINE {$cur[k] = true$}{$mcur[k] := 1;$}
\ELSELINE {$mcur[k] := 0;$ }
\ENDIFLINE
\ENDCASE

\ENDSWITCH

\ENDFOR

\RETURN $cur[idx(\phi)];$

\end{algorithmic}
\end{algorithm}
}

{\small
\begin{algorithm}[t]
\caption{$Iter(\rho, i, \phi, prev, cur, mprev, mcur)$} 
\label{Iter_Algorithm}
\begin{algorithmic}
\REQUIRE {$i > 1.$}

\STATE $prev := cur;$ $mprev := mcur;$
\FORLINE {$k = 1$ to $m$}{$mcur[k] := 0;$}
\FOR {$k = 1$ to $m$}

\SWITCH{$\phi_k$}
\CASELINE {$\bot$}{$cur[k] := false$;}

\CASELINE {$p(\vec c)$}{$cur[k] := p(\vec c) \in \pi_i$;}

\CASELINE {$\neg \psi$}{$cur[k] := \neg cur[idx(\psi)];$}

\CASELINE {$P(\vec c)$}{$cur[k] := cur[idx(\phi_P(\vec c))];$ \COMMENT{Suppose $P(\vec x) := \phi_P(\vec x).$}}

\CASELINE {$\psi_1 \lor \psi_2$}{$cur[k] := cur[idx(\psi_1)] \lor cur[idx(\psi_2)];$}

\CASELINE{$\prev \psi$}{$cur[k] := prev[idx(\psi)];$}

\CASELINE{$\diamonddot \psi$}{$curr[k] := prev[idx(\psi)] \lor prev[\diamonddot \psi];$}

\CASELINE{$\psi_1 \since \psi_2$}{$cur[k] := cur[idx(\psi_2)] \lor (cur[idx(\psi_1)] \land prev[idx(\psi_2)]);$}

\CASELINE{$\prev_n \psi$}{$cur[k] := prev[\psi] \land (\tau_i - \tau_{i-1} < n);$}

\CASE {$\diamonddot_n \psi$}
\STATE $l := prev[idx(\psi)] \land (\tau_i - \tau_{i-1} < n);$
\STATE $r := prev[idx(\diamonddot_n \psi)] \land (n - (\tau_i - \tau_{i-1}) \geq mprev[k]));$
\STATE $cur[k] := l \lor r;$
\IFLINE {$l$}{$mcur[k] := \tau_i - \tau_{i-1} + 1;$}
\ELSIFLINE {$r$}{$mcur[k] := mprev[k] + \tau_i - \tau_{i-1};$}
\ELSELINE{$mcur[k] := 0;$}
\ENDIFLINE
\ENDCASE

\CASE {$\psi_1 \since_n \psi_2$}
\STATE $l := cur[idx(\psi_2)];$
\STATE $r := cur[idx(\psi_1)] \land prev[k] \land (n - (\tau_i - \tau_{i-1}) \geq mprev[k]);$
\STATE $cur[k] := l \lor r;$
\IFLINE {$l$}{$mcur[k] := 1;$}
\ELSIFLINE {$r$}{$mcur[k] := mprev[k] + \tau_i - \tau_{i-1};$}
\ELSELINE{$mcur[k] := 0;$ }
\ENDIFLINE
\ENDCASE

\ENDSWITCH
\ENDFOR
\RETURN $cur[idx(\phi)];$
\end{algorithmic}
\end{algorithm}
}

\begin{theorem}
$(\rho, i) \models \phi$ iff $Monitor(\rho,i,\phi)$ returns true.
\end{theorem}

The $Iter$ function only depends on two worlds: $\rho_i$ and $\rho_{i-1}$, so the
algorithm is trace-length independent. In principle there is no upperbound
to its space complexity, as the timestamp $\tau_i$ can grow arbitrarily large, as
is the case in \cite{Basin11RV}. Practically, however, the timestamps
in Android are stored in a fixed size data structure, so in such a case, when the
policy is fixed, the space complexity is constant (i.e., independent of the length of
history $\rho$).

\section{Examples}
\label{examples}

We provide some basic policies as an example of how we can use this
logic to specify security policies. 
From now on, we shall only quantify over the domain $app$, so in the following
we shall omit the sort annotation in the existential quantifier. 
The predicate $trans$ is the recursive predicate defined in Equation~(\ref{eq:trans})
in the introduction. 
The constant $sink$ 
denotes a service or resource that 
an unprivileged application tries to access via privilege escalation
e.g. send SMS, or access to internet. 
The constant $contact$ denotes the Contact provider app in Android.
We also assume the following ``static'' predicates (i.e., their truth
values do not vary over time):

\begin{itemize}
\item $system(x)$: $x$ is a system app or process.
\item $hasPermissionToSink(y)$: $y$ has permission to access the sink.
\item $trusted(x)$: $x$ is an app that the user trusts. 
This is not a feature of Android, rather, it is a specific feature
of our implementation. We build into our implementation 
a `trust' management app to allow the user a limited control over
apps that he/she trusts. 
\end{itemize}

The following policies refer to access patterns that are 
forbidden. So given a policy $\phi$, the monitor at each state $i$ 
make sure that
$(\rho,i) \not \models \phi$ holds. Assuming that $(\rho, i) \not \models \phi$, where
$i = |\rho|$, holds, then whenever a new event (i.e., the IPC call) $e$ is
registered at time $t$,  
the monitor 
checks that $(([\pi;e], [\tau;t]), i+1) \not \models \phi$ holds.
If it does, then the call is allowed to proceed. Otherwise, it will be terminated. 

\begin{enumerate}
\item $\exists x. (call(x,sink) \wedge \neg system(x) \wedge \neg trusted(x)).$ 

This is a simple policy where we block a direct call from any untrusted
application to the sink.
This policy can serve as a
privilege manager where we dynamically revoke permission for
application to access the sink regardless of the static permission it
asked during installation.

\item $\exists_{x} (trans(x, sink) \land \neg system(x) \land \neg hasPermissionToSink(x)).$

This policy says that transitive calls to a sink from non-system apps are
forbidden, unless the source of the calls already has permission
to the sink. This is then a simple privilege escalation detection (for
non-system apps).

\item $\exists_{x} (trans(x, sink) \wedge \neg system(x) \wedge \neg trusted(x)).$

This is further refinement to the policy in that we also give the user
privilege to decide for themselves dynamically whether or not to trust
an application. 
Untrusted apps can not make transitive call to the sink, but
trusted apps are allowed, regardless of their permissions.

\item 
$\exists_{x} (trans(x, internet) \wedge \neg system(x) \wedge \neg trusted(x) \wedge \diamonddot(call(x, contact))).$

This policy allows privilege escalation by non-trusted apps 
as long as there is no potential
for data leakage through the sink. That is, as soon as a non-system
and untrusted app accesses contact, it will be barred from 
accessing the internet. Note that 
the use of non-metric operator $\diamonddot$ ensures that 
the information that a particular app has accessed contact is persistent.

\end{enumerate}

\section{Implementation}
\label{implementation}

\begin{table}[t]
\caption{Performance Table}
\label{performanceTable}
\centering
	\small{
	\begin{tabular}{ c c c c}
		\hline\hline
		Policy  & Uncached & Cached\\
		\hline
1  & 766.4 & 143.6\\
2  & 936.5 & 423.6\\
3  & 946.8 & 418.3\\
4  & 924.3 & 427.5\\
No Monitor & 758 & 169\\
		\hline
\end{tabular}
	}
\end{table}

\begin{figure}[t]
	\centering
	\subfloat[]{\includegraphics[trim=0 0 0 15, clip, width=0.5\textwidth]{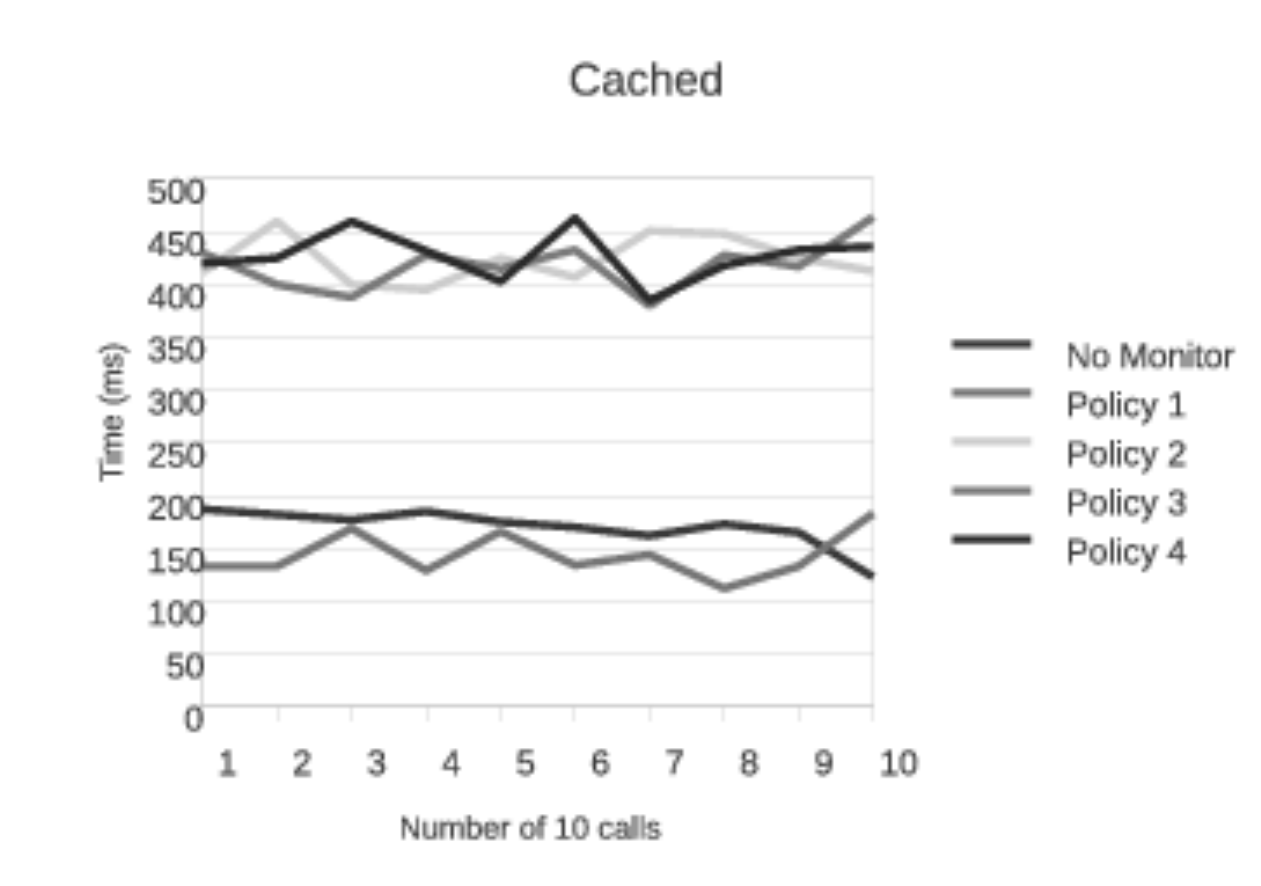}}
	\subfloat[]{\includegraphics[trim=0 0 0 15, clip, width=0.5\textwidth]{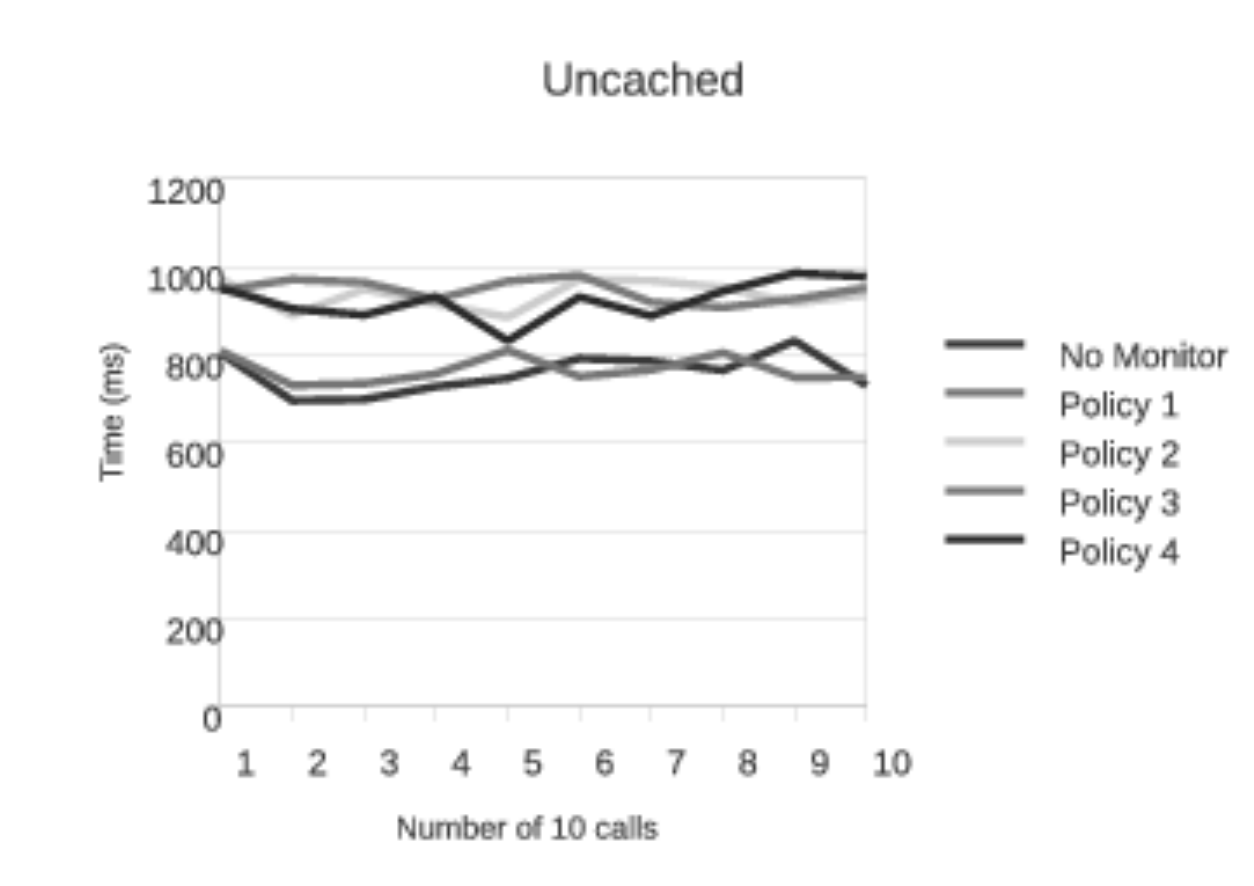}}
	\caption{Timing of Calls}
	\label{callTiming}
\end{figure}

We have implemented the monitoring algorithm presented in the previous section
in Android 4.1. Some modifications to the application framework and the underlying
Linux kernel in Android are neccessary to ensure our monitor can effectively monitor
and stop unwanted behaviours. We have tested our implementation in both Android
emulator and an actual device (Samsung Galaxy Nexus phone). 
We give a brief overview of our implementation; details of the monitor code
generators and the modified Android kernel images are available from
the authors' website. 

Our implementation consists of two parts: the codes that
generate a monitor given a policy specification,  and the modifications of Android
framework and its Linux kernel to hook our monitor and to intercepts IPCs and
access to Android resources. 
To improve runtime performance, the monitor generation is done outside Android; it produces
C codes that are then compiled into a kernel module, and inserted into 
Android boot image. 

The monitor generator takes an input policy, encoded in an XML
format extending that of RuleML.
The monitor generator works by following the logic of the monitoring
algorithm presented in Section~\ref{monitor}. It takes a policy
formula $\phi$, and generates the required data structures and determine
an ordering between elements of $Sub^*(\phi)$ as described earlier,
and produces the codes illustrated in Algorithm~\ref{Init_Algorithm},
\ref{Iter_Algorithm} and \ref{Monitor_Algorithm}. 

The main body of our monitor lies in the Linux kernel space 
as a kernel module.
The reason for this is that there are some cases where
Android leaves the permission checking to the Linux kernel layer
e.g., for opening network socket.
However, to monitor the IPC events between Android components and apps,
we need to place a hook inside the application framework. The IPC between
apps is done through passing a data structure called {\em Intent}, which
gets broken down into {\em parcels} before they are passed down to the
kernel level to be delivered. So intercepting these parcels and 
reconstructing the original Intent object in the kernel space would be more difficult 
and error prone. 
The events generated by apps or components will be passed down to the 
monitor in the kernel, along with the application's user id. 
If the event is a call to the sink, then depending on the policy that is
implemented in the monitor, it will decide to whether block or allow the
call to proceed.  We do this through our custom additional system 
calls to the Linux kernel which goes to this monitor.

Our implementation places hooks in four services, 
namely accessing internet, sending SMS, 
accessing location, and accessing contact database. 
For each of this sink, we add a virtual
UID in the monitor and treat it as a component of Android. 
We currently track only IPC calls through the Intent passing mechanism. 
This is obviously not enough to detect all possible communications between
apps, e.g., those that are done through file systems, or side channels, such
as vibration setting (e.g., as implemented in SoundComber~\cite{Soundcomber}), so our
implementation is currently more of a proof of concept.  
In the case of SoundComber, our monitor can actually intercepts the calls
between colluding apps, due to the fact that they utilise intent broadcast
through IPC to synchronize data transfer.

We have implemented some apps to test policies we mentioned in Section~\ref{examples}. 
In Table \ref{performanceTable} and Figure \ref{callTiming}, we provide some measurement of
the timing of the calls between applications. The policy numbers in Table~\ref{performanceTable}
refer to the policies in Section~\ref{examples}. To measure the average time
for each IPC call, we construct a chain of ten apps, making successive calls between
them, and measure the time needed for one end to reach the other. 
We measure two different average timings in miliseconds (ms)
for different scenarios, based on whether the apps are in the background cache (i.e.,
suspended) or not. 
We also measure time spent
on the monitor actually processing the event, which are around 1 ms
for policy 1, and around 10 ms for the other three policies. 
This shows that the  time spent in processing the
event is quite low, but more overhead comes from the space required to
process the event (there is a big jump in overall timing from simple
rules with at most 2 free variables to more complex one with 3 free
variables). Figure \ref{callTiming} 
shows
that the timing of calls over time for
each policy are roughly the same. 
This backs our claim that
even though our monitor implements history-based access control, its 
performance does not depend on the size of the history.

\section{Conclusion, related and future work}
\label{conclusion}

We have shown a policy language design based on MTL that can effectively
describe various scenarios of privilege escalation in Android. Moreover,
any policy written in our language can be effectively
enforced. The key to the latter is the fact that our enforcement procedure
is trace-length independent. We have also given a proof-of-concept implementation
on actual Android devices and show that our implementation can effectively
enforce RMTL policies. 

We have already discussed related work in runtime monitoring based on LTL in the introduction.
We now discuss briefly related work in Android security. 
There is a large body of works in this area, more than what can be
reasonably surveyed here, so we shall focus on the most relevant
ones to our work, i.e., those that deal with privilege
escalation. For a more comprehensive survey on
other security extensions or analysis, the interested reader can
consult~\cite{LastPE}. 
QUIRE \cite{QUIRE} is an application centric approach 
to privilege escalation, done by tagging the intent
objects with the caller's UID. Thus, 
the recipient application can check the permission of the source of the
call chain. 
IPC Inspection \cite{IPCInspection} is another application centric solution 
that works 
by reducing the privilege of the
recipient application when it receives a communication from a less
privileged application. 
XManDroid \cite{LastPE} is a system centric solution, just like ours. Its security monitor
maintains a call graph between apps. It is the closest to our solution,
except that 
we are using temporal logic to specify a policy, and our policy can be 
modified modularly.
Our policy language is also more expressive as we can specify both
temporal and metric propertes. 
TaintDroid~\cite{TaintDroid} is another system-centric solution, but it is 
designed to track data flow, rather than control flow, via taint analysis, so
privilege escalation can be inferred from leakage of data.

We currently do not deal with quantifiers directly in our algorithm.
Such quantifiers are expanded into purely propositional connectives (when the domain
is finite), which is exponential in the number of variables in the policy. As an immediate future work,
we plan to investigate whether techniques using {\em spawning automata}~\cite{bauer2013rv}
can be adapted to our setting to allow a ``lazy'' expansion of quantifiers as needed. 
It is not possible to design trace-length-independent
monitoring algorithms in the unrestricted first-order LTL~\cite{bauer2013rv}, so
the challenge here is to find a suitable restriction that can
be enforced efficiently.

\bibliographystyle{plain}
\bibliography{biblio}

\newpage
\newpage
\appendix

\section{Proofs}


\noindent
\textbf{Lemma \ref{lm:order}.}
\textit{
The relation $\prec$ on RMTL formulas is a well-founded partial order. 
}
\begin{proof}
We first show that $\prec$ is well-founded. Suppose otherwise: then there is 
an infinite descending chain of formulas:
$$
\cdots \prec_S \phi_n \prec_S \phi_{n-1} \prec_S \cdots \prec_S \phi_2 \prec_S \phi_1.
$$
Obviously, none of $\phi_i$'s can be a bottom element (i.e., those that take the
form as specified in clause (1) of Definition~\ref{def:order}). 
Furthermore, there must be an $i$ such that $\phi_i = P(\vec t)$ and $\phi_{i+1} = \phi_P(\vec t)$
where $P$ is a recursive predicate defined by $P(\vec x) := \phi_P(\vec x).$ 
If no such $i$ exists
then all the instances of the relation $\prec_S$ in the chain must be instances
of clause (3) and (4) in Definition~\ref{def:order}, and the chain would be finite
as those two clauses relate only strict subformulas. 
So without loss of generality, let us assume that $\phi_1 = P(\vec t)$. 
We claim that for every $j > 1$, every occurrence of any recursive predicate in $\phi_j$ is guarded.
We prove this by induction on $j$. If $j=2$ then we have $\phi_2 \prec_S P(\vec t)$. In this case, $\phi_2$ must be $\phi_P(\vec t)$, and by the guardedness condition,
all recursive predicates in $\phi_P$ are guarded.
If $j > 2$, then we have $\phi_j \prec_S \phi_{j-1}$. By induction hypothesis,
all recursive predicates in $\phi_{j-1}$ are guarded. In this case, 
the relation $\phi_j \prec_S \phi_{j-1}$ must be an instance of either
clause (3) or clause (4) of Definition~\ref{def:order}, and therefore $\phi_j$ also satisfies the guardedness
condition. 

So now we have that none of $\phi_i$'s are recursive predicates. This means
that all instances of $\prec_S$ in the chain must be instances of clause (3) and (4)
in Definition~\ref{def:order}, and consequently the size of the formulas in
the chain must be strictly decreasing. Thus the chain cannot be infinite, contrary
to the assumption. 

Anti-symmetry follows immediately from well-foundedness. Suppose $\prec$ is not
anti-symmetric. Then we have a chain
$$
\phi = \phi_1 \prec_S \phi_2 \prec_S \phi_3 \prec_S \cdots \prec_S \phi_n = \phi
$$
where $n > 1.$
We can repeat this chain to form an infinite descending chain, which contradicts
the well-foundedness of $\prec.$
\qed
\end{proof}

\noindent
\textbf{Lemma \ref{lm:minimal} (Minimality).}
\textit{
If $(\rho, i) \models \phi_1 \since_n \phi_2$ ($(\rho, i) \models \diamonddot_n \phi$) then there exists an $m \leq n$ 
such that 
$(\rho, i) \models \phi_1 \since_m \phi_2$ (resp. $(\rho, i) \models \diamonddot_m \phi$) and 
such that for every $k$ such that $0 < k < m$, 
we have $(\rho, i) \not \models \phi_1 \since_k \phi_2$ (resp., 
$(\rho, i) \not \models \diamonddot_k \phi$).
}
\begin{proof}
We show a case for $(\rho, i) \models \phi_1 \since_n \phi_2$; the other case is straightforward.
We prove this by induction on $i.$
\begin{itemize}
\item Base case: $i = 1.$ Since $(\rho, i) \models \phi_1 \since_n \phi_2$, it must be the case that
$(\rho, i) \models \phi_2$. In this case, let $m = 1.$ Obviously $(\rho, i) \models \phi_1 \since_m \phi_2$
and $m$ is minimal. 

\item Inductive case: $i > 1.$
We have $(\rho, i) \models \phi_1 \since_n \phi_2$. By the definition of $\models$, there exists
$j \leq i$ such that $(\rho, j) \models \phi_2$ and 
$(\rho, k) \models \phi_1$ for every $k$ s.t.
$j < k \leq i$ and $\tau_i - \tau_j < n.$
If $(\rho, i) \models \phi_2$, then 
we have $(\rho, i) \models \phi_1 \since_1 \phi_2.$
In this case, let $m = 1.$
If $(\rho, i) \not \models \phi_2$, 
then it must be the case that $j < i$. 
It is not difficult to see that in this case we must have
$$(\rho, i-1) \models \phi_1 \since_{n - (\tau_i - \tau_{i-1})} \phi_2.$$
By the induction hypothesis, we have there is an $m'$ such that
$$
(\rho, i-1) \models \phi_1 \since_{m'} \phi_2
$$
and for every $l$ s.t. $l < m'$, we have
$(\rho, i-1) \not \models \phi_1 \since_l \phi_2.$
In this case, we let $m = m' + \tau_i - \tau_{i-1}.$
It is straightforward to check that
$(\rho, i) \models \phi_1 \since_m \phi_2.$

Now, we claim that this $m$ is minimal. Suppose otherwise, i.e.,
there exists $k < m$ s.t. 
$(\rho, i) \models \phi_1 \since_k \phi_2.$
Since $(\rho, i) \not \models \phi_2$, we must have
$(\rho, i-1) \models \phi_1 \since_{k - (\tau_i - \tau_{i-1})} \phi_2$.
But $(k - (\tau_i - \tau_{i-1})) < m'$, so this contradicts
the minimality of $m'.$

\end{itemize}
\qed
\end{proof}

\begin{lemma}[Monotonicity]
\label{lm:monotonicity}
If $(\rho, i) \models \phi_1 \since_n \phi_2$ (resp., $(\rho, i) \models \prev_n \phi$
and $(\rho, i) \models \diamonddot_n \phi$) then for every $m \geq n$,
we have $(\rho, i) \models \phi_1 \since_m \phi_2$ (resp., 
$(\rho, i) \models \prev_m \phi$
and $(\rho, i) \models \diamonddot_m \phi$).
\end{lemma}
\begin{proof}
Straightforward from the definition of $\models.$ \qed
\end{proof}

\noindent \textbf{Theorem \ref{thm:recursive-form} (Recursive forms).}
\textit{
For every model $\rho$, every $n \geq 1$, $\phi$, $\phi_1$ and $\phi_2$, 
and every $1 < i \leq |\rho|$, the following hold:}
\begin{enumerate}
\item $(\rho, i) \models \phi_1 \since_n \phi_2$ iff
$(\rho, i) \models \phi_2$, or 
$(\rho, i) \models \phi_1$ and 
$(\rho, i-1) \models \phi_1 \since_n \phi_2$ and 
$n - (\tau_i - \tau_{i-1}) \geq \mn(\rho, i-1, \phi_1 \since_n \phi_2).$

\item $(\rho, i) \models \diamonddot_n \phi$ iff
$(\rho, i - 1) \models \phi$ and $\tau_i - \tau_{i-1} < n$, or 
$(\rho, i-1) \models \diamonddot_n \phi$ and 
$n - (\tau_i - \tau_{i-1}) \geq \mn(\rho, i-1, \diamonddot_n \phi).$

\end{enumerate}
\begin{proof}
We show the case for $\since_n$; the other case is similar.

Suppose $(\rho, i) \models \phi_1 \since_n \phi_2$. By definition, there exists $j \leq i$ such that
\begin{equation}
\label{eq:re1}
(\rho, j) \models \phi_2,
\end{equation}
\begin{equation}
\label{eq:re2}
\tau_i - \tau_j \leq n, 
\end{equation}
and for every $k$ s.t. $j < k \leq i$ we have 
\begin{equation}
\label{eq:re3}
(\rho, k) \models \phi_2, 
\end{equation}
Suppose that the right-hand side of iff doesn't hold, i.e., we have 
$(\rho, i) \not \models \phi_2$, and
that one of the following hold:
\begin{itemize}
\item $(\rho, i) \not \models \phi_1$, 
\item $(\rho, i-1) \not \models \phi_1 \since_n \phi_2$, or
\item $n - (\tau_i - \tau_{i-1}) < \mn(\rho, i-1,  \phi_1 \since_n \phi_2).$
\end{itemize}
The first case contradicts our assumption in (\ref{eq:re3}), so it cannot hold.
Note that since $(\rho, i) \not \models \phi_2$, it must be the case that
$j \leq i-1$, so (\ref{eq:re1}), (\ref{eq:re2}), and (\ref{eq:re3}) above entail that
$(\rho, i-1) \models \phi_1 \since_n \phi_2$, so the second case can't hold either. 
For the third case: from (\ref{eq:re2}) we have 
$\tau_i - \tau_j = (\tau_i - \tau_{i-1}) + (\tau_{i-1} - \tau_j) \leq n$
so
$$
\tau_{i-1} - \tau_j \leq n - (\tau_{i} - \tau_{i-1}).
$$
This, together with (\ref{eq:re1}) and (\ref{eq:re2}), implies that
$(\rho, i-1) \models \phi_1 \since_{n - (\tau_{i} - \tau_{i-1})} \phi_2$.
If $n - (\tau_i - \tau_{i-1}) < \mn(\rho, i-1, \phi_1 \since_n \phi_2)$ holds, then it would
contradict the fact that $\mn(\rho, i-1, \phi_1 \since_n \phi_2)$ is the minimal
index as guaranteed by Lemma~\ref{lm:minimal}. Hence the third case cannot hold either.

For the other direction, suppose that either of the following holds:
\begin{itemize}
\item $(\rho, i) \models \phi_2$, or
\item $(\rho, i) \models \phi_1$ and 
$(\rho, i-1) \models \phi_1 \since_n \phi_2$ and 
$n - (\tau_i - \tau_{i-1}) \geq \mn(\rho, i-1, \phi_1 \since_n \phi_2).$
\end{itemize}
If it is the first case, i.e., $(\rho,i)\models \phi_2$, then trivially $(\rho, i) \models \phi_1 \since_n \phi_2.$
So suppose it is the second case. 
By Lemma~\ref{lm:minimal}, we have $(\rho, i-1) \models \phi_1 \since_m \phi_2$, where $m = \mn(\rho, i-1, \phi_1 \since_n \phi_2).$
Since $n - (\tau_i - \tau_{i-1}) \geq m$, by Lemma~\ref{lm:monotonicity}, we have $(\rho, i-1) \models \phi_1 \since_{n - (\tau_i - \tau_{i-1})} \phi_2.$
This means there exists $j \leq i-1$ such that 
\begin{itemize}
\item $(\rho, j) \models \phi_2$,
\item $(\tau_{i-1} - \tau_j) \leq n - (\tau_i - \tau_{i-1})$, and
\item $(\rho, k) \models \phi_1$, for every $k$ s.t. $j < k \leq i-1$.
\end{itemize}
The second item implies that $(\tau_i - \tau_j) \leq n$.
These and the fact that $(\rho, i) \models \phi_1$ entail that $(\rho, i) \models \phi_1 \since_n \phi_2.$
\qed
\end{proof}

\noindent \textbf{Theorem 2.}
\textit{
$(\rho, i) \models \phi$ iff $Monitor(\rho,i,\phi)$ returns true.
}
\begin{proof}
The proof is quite straightforward, because each step in the calculation of the truth values of
subformulas of $\phi$ corresponds to their (recursive-form) semantics reading and Theorem~\ref{thm:recursive-form}. 
Moreover, by the well-foundedness of $\prec$, this algorithm terminates, and at each step, the calculation
of a subformula $\phi$ uses only values of other subformulas smaller than $\phi$ in the ordering $\prec$, or
truth values of subformulas in the previous world (which are already defined). 
In the case of $\since_n$, the updates of the value $mcur$ and $mprev$ corresponds
exactly to the construction shown in the proof of Lemma~\ref{lm:minimal}. 
\qed
\end{proof}

\end{document}